\documentclass[11pt,english]{article}
\usepackage{color}
\usepackage{amsmath}
\usepackage{amsthm}
\usepackage{amssymb}
\usepackage{amsmath}
\usepackage{mathcomp}
\usepackage{dsfont}
\usepackage[toc,page]{appendix}
\usepackage{graphicx} 
\usepackage{subfigure}
\usepackage{enumerate}

\usepackage{hyperref}

\newtheorem{definition}{Definition}[section]
\newtheorem{lemma}[definition]{Lemma}

\newtheorem{theorem}[definition]{Theorem}

\newtheorem{corollary}[definition]{Corollary}

\numberwithin{equation}{section}

\def\tr{\mathrm{tr}}

\def\a0{\mathfrak{a}_0}

\def\cF{\mathcal{F}}
\def\ph{\varphi}
\def\bN{\mathbb{N}}
\def\bR{\mathbb{R}}
\def\bP{\mathbb{P}}
\def\bZ{\mathbb{Z}}
\def\cO{\mathcal{O}}
\def\cF{\mathcal{F}}
\def\cN{\mathcal{N}}
\def\cG{\mathcal{G}}
\def\cM{\mathcal{M}}
\def\cJ{\mathcal{J}}
\def\cC{\mathcal{C}}
\def\cQ{\mathcal{Q}}
\def\cV{\mathcal{V}}
\def\cE{\mathcal{E}}
\def\wt{\widetilde}

\begin{document}

\title{Central Limit Theorem for Bose-Einstein Condensates}

\author{Simone Rademacher, Benjamin Schlein \\
\\
Institute of Mathematics, University of Zurich\\
Winterthurerstrasse 190, 8057 Zurich, 
Switzerland}

\maketitle

\begin{abstract}
We consider a Bose gas trapped in the unit torus in the Gross-Pitaevskii regime. In the ground state, we prove that fluctuations of bounded one-particle observables satisfy a central limit theorem.
\end{abstract}

\section{Introduction}

We consider Bose gases consisting of $N$ particles trapped in the three dimensional unit torus $\Lambda$ interacting through a repulsive potential with scattering length of the order $N^{-1}$ (Gross-Pitaevskii regime). The Hamilton operator is given by 
\begin{equation}\label{eq:HN} H_N = \sum_{j=1}^N -\Delta_{x_j} + \sum_{i<j}^N N^2 V(N(x_i -x_j)) \end{equation}
and, according to the bosonic statistics, it acts on $L^2_s (\Lambda^N)$, the subspace of $L^2 (\Lambda^N)$ consisting of functions that are symmetric with respect to permutations of the $N$ particles. 

At zero temperature, the system relaxes to the ground state, described by a normalized eigenvector of (\ref{eq:HN}) associated with its smallest eigenvalue. From \cite{LS,LS2,NRS,BBCS}, the ground state of (\ref{eq:HN}) is known to exhibit complete condensation in the zero-momentum mode $\ph_0$ defined through $\ph_0 (x) = 1$ for all $x \in \Lambda$. In other words, if we denote by $\gamma^{(1)}_N = \tr_{2,\dots, N} |\psi_N \rangle \langle \psi_N|$ the one-particle reduced density matrix associated with the normalized ground state vector $\psi_N$ (without loss of generality, we choose $\psi_N$ to be the unique positive normalized ground state vector of $H_N$), then (for example, in the trace-norm topology) 
\begin{equation}\label{eq:BEC} \lim_{N \to \infty} \gamma^{(1)}_N = |\ph_0 \rangle \langle \ph_0| \end{equation}
Eq. (\ref{eq:BEC}) states that almost all particles, up to a fraction vanishing as $N \to \infty$, are in the same one-particle state $\ph_0$. In fact, (\ref{eq:BEC}) also implies convergence of higher order reduced densities. For $k \in \bN$, we define the $k$-particle reduced density $\gamma^{(k)}_N = \tr_{k+1, \dots , N} |\psi_N \rangle \langle \psi_N|$ (normalized so that $\tr \, \gamma^{(k)}_N = 1$). Then, we find \begin{equation}\label{eq:BECk} \lim_{N \to \infty} \gamma^{(k)}_{N} = |\ph_0 \rangle \langle \ph_0|^{\otimes k}\end{equation} 
for all fixed $k \in \bN$. It should be stressed, however, that (\ref{eq:BEC}) and (\ref{eq:BECk}) do not imply that $\ph_0^{\otimes N}$ is a good approximation for the ground state vector $\psi_N$. From \cite{LY,LSY}, it is known that the ground state energy per particle is such that 
\begin{equation}\label{eq:gs-en}  \lim_{N \to \infty} N^{-1} \langle \psi_N, H_N \psi_N \rangle = 4\pi \frak{a}_0   \end{equation}
where $\frak{a}_0$ denotes the scattering length of $V$. A simple computation shows instead that 
\begin{equation}\label{eq:en-fact} \lim_{N\to \infty} N^{-1} \langle \ph_0^{\otimes N} , H_N \ph_0^{\otimes N} \rangle = \frac{\widehat{V} (0)}{2} \end{equation}
Since $\widehat{V} (0) > 8 \pi a_0$, the energy of the factorized state $\ph_0^{\otimes N}$ is always much larger than the energy of the ground state $\psi_N$; the excess energy is macroscopic, of order $N$. In the Gross-Pitaevskii regime, correlations among particles are crucial. They are responsible for lowering the energy from (\ref{eq:en-fact}) to (\ref{eq:gs-en}). Since they vary on the length-scale $N^{-1}$, they disappear in the limit $N \to \infty$, but only at the level of the reduced densities. 

Eq. (\ref{eq:BECk}) allows us to estimate averages of one-particle observables. Given a self-adjoint operator $O$ on $L^2 (\Lambda)$ (a one-particle operator), we denote by $O_j = 1 \otimes \dots \otimes 1 \otimes O \otimes 1 \otimes \dots \otimes 1$ the operator on $L^2_s (\Lambda^N)$ acting as $O$ on the $j$-th particle, and as the identity elsewhere. From (\ref{eq:BEC}), we find 
\[ \lim_{N \to \infty} \Big\langle \psi_N, \frac{1}{N} \sum_{j=1}^N O_j  \psi_N \Big\rangle = \lim_{N \to \infty} \tr \, \gamma_N^{(1)} O = \langle \ph_0 , O \ph_0 \rangle \]
In fact, Eq. (\ref{eq:BECk}) with $k =2$ also implies a law of large numbers for the probability distribution associated with the ground state wave function $\psi_N$. For any $\delta > 0$, 
we find 
\begin{equation}\label{eq:lln}  \lim_{N \to \infty} \bP_{\psi_N} \Big( \Big| \frac{1}{N} \sum_{j=1}^N O_j - \langle \ph_0, O \ph_0 \rangle \Big|  \geq \delta \Big) = 0 \end{equation}
To prove (\ref{eq:lln}), we set $\widetilde{O} = O - \langle \ph_0 , O \ph_0 \rangle$ and we use Markov's inequality to estimate  
\[\begin{split} \bP_{\psi_N} \Big( \Big| \frac{1}{N} \sum_{j=1}^N O_j - \langle \ph_0, O \ph_0 \rangle \Big|  \geq \delta \Big) &= \bP \Big( \frac{1}{N\delta} \Big| \sum_{j=1}^N \wt{O}_j \Big| \geq 1 \Big)  \\ &\leq \frac{1}{N^2 \delta^2} \langle \psi_N, \sum_{i,j =1}^N \widetilde{O}_i \widetilde{O}_j \psi_N \rangle \\ &= \frac{N(N-1)}{2N^2 \delta^2} \, \tr \; \gamma^{(2)}_N \widetilde{O}_1 \otimes \widetilde{O}_2 + \frac{1}{N \delta^2} \, \tr  \, \gamma_N^{(1)} \widetilde{O}^2 \to 0 \end{split} \]
as $N \to \infty$, because $\tr \, |\ph_0 \rangle \langle\ph_0| \widetilde{O} = 0$. Thus, the correlation structure characterizing the ground state wave function $\psi_N$ does 
not affect the law of large numbers. As $N \to \infty$, the average concentrates around the value $\langle \ph_0, O \ph_0 \rangle$, exactly as it would in the completely factorized 
state $\ph_0^{\otimes N}$. 

In our main theorem, we show that the probability distribution generated by $\psi_N$ also satisfies a (multivariate) central limit theorem; (joint) fluctuations of observables of averages of the form $\sum_{j=1}^N O_j$ are Gaussian and have size of order $\sqrt{N}$. A similar question has been addressed in \cite{BKS,BSS}, for the time-evolution of mean-field quantum systems.

\begin{theorem}\label{thm:main}
Let $V \in L^3( \mathbb{R}^3)$ be non-negative, spherically symmetric and compactly supported. Let $\psi_N \in L_s^2 ( \Lambda^N)$ with $\| \psi_N \| =1$ denote the (unique) positive normalized ground state wave function of the Hamiltonian \eqref{eq:HN}.

For $k \in \mathbb{N}$, let $O_1, \dots, O_k$  be bounded operators on $L^2( \Lambda )$. For $j=1,\dots, k$, let 
\begin{equation}
\label{eq:cO}
\cO_j = \frac{1}{\sqrt{N}} \sum_{\ell =1}^N \left(O_{j,\ell} - \langle \ph_0, O_j \ph_0 \rangle \right)
\end{equation} 
where $O_{j,\ell} = 1 \otimes \dots \otimes O_j \otimes \dots \otimes 1$ denotes the operator acting as $O_j$ on the $\ell$-th particle and as the identity on the other $(N-1)$ particles. 

For $j=1, \dots , k$, let  
\[ 
\nu_j (p) = (\widehat{q_0 O_j \ph_0}) (p) \cosh (\mu_p)+ (\widehat{\overline{q_0 O_j \ph_0}}) (p) \sinh (\mu_p) 
\]
where $\ph_0 \in L^2 (\Lambda)$ is the zero-momentum mode (ie. $\ph_0 (x) = 1$ for all $x \in \Lambda$), $q_0 = 1- |\ph_0\rangle \langle \ph_0|$ and 
\begin{equation}\label{eq:mup0} \mu_p = \frac{1}{4} \log \left( \frac{p^2}{p^2 + 16 \pi \frak{a}_0} \right) \end{equation}
for all $p \in \Lambda^*_+ = 2 \pi \mathbb{Z}^3 \setminus \lbrace 0 \rbrace$. Let \begin{equation}
\label{eq:sigma}
\Sigma_{i,j} = \left\{ \begin{array}{ll} \langle \nu_i, \nu_j \rangle = \sum_{p \in \Lambda^*_+} \overline{\nu_i (p)}  \, \nu_j (p) \quad &\text{if $i \leq j$} \\
\langle \nu_j, \nu_i \rangle \quad &\text{if $j < i$} \end{array} \right. 
\end{equation} 
We assume the $k \times k$ matrix $\Sigma = (\Sigma_{i,j} )_{1 \leq i,j \leq k}$ to be invertible (by definition $\text{Re } \Sigma \geq 0$). 

Let  $g_1, \dots g_k \in L^1( \bR)$ with Fourier transforms $\widehat{g}_j \in L^1 ( \bR, (1+|s|^4)  ds)$ for all $j =1 , \dots , k$. Then  
\begin{align*}
&\left\vert \mathbb{E}_{\psi_N} \left[ g_1 ( \mathcal{O}_1) \cdots g_k( \mathcal{O}_k ) \right] -  \int d\lambda_1 \dots d\lambda_k \; g_1(\lambda_1) \dots g_k(\lambda_k) \; \frac{e^{-1/2 \sum_{i,j=1}^k \Sigma_{i,j}^{-1} \lambda_i \lambda_j}}{\sqrt{( 2 \pi)^k \mathrm{det} \Sigma}}  \right\vert\\
& \hspace{8cm} \leq \frac{C}{N^{1/4}} \prod_{j=1}^k \int ds \; | \widehat{g}_j (s)| (1 + |s|^4 \| O_j \|^4)
\end{align*}
for a constant $C > 0$ depending only on $k \in \bN$. 
\end{theorem}

As a corollary to Theorem \ref{thm:main}, we obtain a Berry-Ess\'{e}en type central limit theorem, with  explicit control of the error. We skip the proof which follows very closely 
the one of \cite[Corollary 1.2]{BSS}.
\begin{corollary}\label{cor:BE}
Let $V \in L^3( \mathbb{R}^3)$ be non-negative, spherically symmetric and compactly supported. Let $\psi_N \in L_s^2 ( \Lambda^N)$ with $\| \psi_N \| =1$ denote the (unique) positive normalized ground state wave function of the Hamiltonian \eqref{eq:HN}. Let $O$ be a bounded, self-adjoint one-particle operator on $L^2 (\Lambda)$ and 
\[ \cO = \frac{1}{\sqrt{N}} \sum_{j=1}^N \left( O - \langle \ph_0 , O \ph_0 \rangle \right) .\] 
For every $-\infty < \alpha < \beta < \infty$, there exists a constant $C>0$ such that
\begin{equation}\label{eq:cor}
\vert \mathbb{P}_{\psi_N} \left( \mathcal{O} \in \left[ \alpha; \beta \right] \right) - \mathbb{P} \left( G \in \left[ \alpha; \beta \right] \right) \vert \leq C N^{-1/8},
\end{equation}
where $G$ is the centred Gaussian random variable  with variance $\| \nu \|^2_2$, where $\nu \in \ell^2 ( \Lambda_+^*)$ is defined through 
\[ 
\nu (p) = (\widehat{q_0 O \ph_0}) (p) \cosh (\mu_p) + (\widehat{\overline{q_0 O \ph_0}}) (p) \sinh (\mu_p) \]
with $\mu$ as defined in (\ref{eq:mup0}). 
\end{corollary}

{\it Remark:} The dependence of the constant $C$ on the r.h.s. of (\ref{eq:cor}) on $\alpha, \beta$ can be controlled by $C \leq c ( 1 + |\beta - \alpha|)$ for a universal constant $c > 0$. 

{\it Remark:} Theorem \ref{thm:main} and Corollary \ref{cor:BE} imply that, although the ground state exhibits important correlations, the probability distribution it generates still satisfies a central limit theorem. In our analysis, this will follow from the bound \eqref{eq:psi-app} taken from \cite{BBCS}, which shows that, up to corrections described by the unitary operator $S$ (which are important to compute the energy but, as we will prove, do not affect the distribution of products of bounded one-particle observables of the form \eqref{eq:cO}), the ground state is approximately quasi-free; in fact, the form of the covariance matrix \eqref{eq:sigma} is entirely determined by the two Bogoliubov transformations $T_\eta$ and $T_\tau$ appearing in \eqref{eq:psi-app}. 

{\it Remark:} While Theorem \ref{thm:main} and Corollary \ref{cor:BE} state properties of the probability distribution associated with the ground state of (\ref{eq:HN}), it is also possible to study the probability distribution generated by low-energy excited states; we briefly address 
this question in Appendix \ref{app}. 

\section{Approximation of ground state}

The proof of Theorem \ref{thm:main} is based on a norm-approximation for the ground state wave function $\psi_N$, recently obtained in \cite{BBCS}. To illustrate this approximation, we need to introduce some unitary operators leading to an approximate diagonalisation of the Hamilton operator (\ref{eq:HN}). 

First of all, following \cite{LNSS} we observe that every $\psi_N \in L^2_s (\Lambda^N)$ can be uniquely decomposed as 
\[ \psi_N = \alpha_0 \ph_0^{\otimes N} + \alpha_1 \otimes_s \ph_0^{\otimes (N-1)} + \dots + \alpha_N \]
with $\alpha_j \in L^2_{\perp \ph_0} (\Lambda)^{\otimes_s j}$, where $\otimes_s$ denotes the symmetrized tensor product and $L^2_{\perp \ph_0} (\Lambda)$ denotes the orthogonal complement in $L^2 (\Lambda)$ of the condensate wave function $\ph_0$. This observation allows us to define the unitary map $U_N : L^2_s (\Lambda^N) \to \cF_+^{\leq N}$ by $U_N \psi_N = \{ \alpha_0, \alpha_1, \dots , \alpha_N \}$. Here 
\[ \cF_+^{\leq N} = \bigoplus_{j=0}^N L^2_{\perp \ph_0} (\Lambda)^{\otimes_s j} \]
is the truncated Fock space (describing states with at most $N$ particles) constructed over $L^2_{\perp \ph_0} (\Lambda)$. On $\cF_+^{\leq N}$, we describe orthogonal excitations of the condensate (applying $U_N$, we factor out the Bose-Einstein condensate and we focus on its excitations). The action of $U_N$ is determined by the rules (see \cite{LNSS}) 
\begin{equation}\label{eq:rules} \begin{split} 
U_N^* a_0^* a_0 U_N &= N- \cN_+ \\
U_N^* a_p^* a_0 U_N &= a_p^* \sqrt{N- \cN_+} =: \sqrt{N} b_p^* \\
U_N^* a_0^* a_p U_N &= \sqrt{N-\cN_+} a_p =: \sqrt{N} b_p \\
U_N^* a_q^* a_p U_N &= a_q^* a_p \end{split} \end{equation}
for all $p,q \in \Lambda^*_+ := 2\pi \bZ^3 \backslash \{ 0 \}$. For $r \in 2\pi \bZ^3$, the operators $a^*_r, a_r$ are the usual creation and annihilation operators, creating and, respectively, annihilating a particle with momentum $r$. They satisfy the canonical commutation relations 
\begin{equation}\label{eq:ccr} \left[ a_r , a_s^* \right] = \delta_{r,s}, \qquad \left[ a_r , a_s \right] = \left[ a_r^* , a_s^* \right] = 0 \, . \end{equation} 

On the r.h.s. of (\ref{eq:rules}), $\cN_+ = \sum_{p \in \Lambda^*_+}  a_p^* a_p$ denotes the number of particles operator on $\cF_+^{\leq N}$; it measure the number of excitations of the Bose-Einstein condensate. Furthermore, for $p \in \Lambda^*_+$, we introduced the modified creation and annihilation operators 
\[ b_p^* = a_p^* \sqrt{\frac{N-\cN_+}{N}}, \qquad b_p = \sqrt{ \frac{N-\cN_+}{N}} a_p \]
These operators create and annihilate excitations, keeping the total number of particles invariant. They are useful because, in contrast to the standard creation and annihilation operators, they leave the Hilbert space $\cF^{\leq N}_+$ invariant and, at the same time, they provide a good approximation of $a_p^*, a_p$ on states exhibiting Bose-Einstein condensation (ie. states with few excitations, where $\cN_+ \ll N$). From (\ref{eq:ccr}), it is easy to derive the commutation relations 
\begin{equation}
\label{eq:ccr-mod}
\begin{split}
\left[ b_p ,b_q^* \right] &= \delta_{p,q} \left( 1 - \frac{\cN_+}{N} \right) - \frac{1}{N} a_q^* a_p , \qquad \left[ b_p , b_q \right] = \left[ b_p^* , b_q^* \right] = 0 
\end{split} 
\end{equation} 
More generally, for $h \in \ell^2 (\Lambda^*_+)$, we can define 
\begin{equation}
\label{eq:bb}
b(h) = \sum_{p \in \Lambda^*_+} \overline{h (p)} b_p, \qquad b^* (h) = \sum_{p \in \Lambda^*_+} h(p) b_p^* \end{equation}
Then $\| b(h) \xi \| \leq \| h \| \| \cN_+^{1/2} \xi \|$ and $\| b^* (h) \xi \| \leq \| h \| \| (\cN_+ + 1)^{1/2} \xi \|$. Here and in the following we use the notation $\| h \| = \big(\sum_{p \in \Lambda^*_+} |h(p)|^2 \big)^{1/2}$ for the $\ell^2$-norm.

Applying $U_N$ to the ground state $\psi_N$, we remove products of the condensate wave function $\ph_0$. Correlations are left in the excitation vector $U_N \psi_N \in \cF_+^{\leq N}$. To obtain a norm-approximation for $\psi_N$, we have to model the correlation structure. To this end, we introduce the generalized Bogoliubov transformation
\begin{equation}\label{eq:BogT} T_\eta = \exp \left[ \frac{1}{2} \sum_{p \in \Lambda^*_+} \eta_p \left( b^*_p b^*_{-p} - b_{-p} b_p \right) \right]  \end{equation}
where the coefficients $\eta_p$ are defined for $p \in \Lambda_+^*$ through \begin{equation*} 
\eta_p = - \frac{1}{N^2} \widehat{w}_\ell (p/N) \end{equation*}
with 
\[ \widehat{w}_\ell (k) = \int_{\bR^3} w_\ell (x) e^{-i k \cdot x } dx \]
and $w_\ell (x) = 1- f_\ell (x)$. Here we denote by $f_\ell$ the ground state solution of the Neumann problem 
\[ \left[- \Delta + \frac{1}{2} V \right] f_\ell = \lambda_\ell f_\ell \]
on the ball $|x| \leq N\ell$, for an $\ell \in (0;1/2)$. We recall from \cite[Lemma 3.1]{BBCS} that
\begin{equation}\label{eq:etap} |\eta_p| \leq C |p|^{-2} \end{equation} 
and that 
\begin{equation}\label{eq:Vfell-int} \left| \int V(x) f_\ell (x) dx - 8\pi \frak{a}_0 \right| \leq C N^{-1} \end{equation}

As proven in \cite[Lemma 2.3]{BBCS}, the action of the generalized Bogoliubov transformation (\ref{eq:BogT}) on (modified) creation and annihilation operators is given, for any $p \in \Lambda^*_+$, by the formulas
\begin{equation}\label{eq:act-T} \begin{split} T_\eta^* b_p T_\eta &= \cosh (\eta_p ) b_p + \sinh (\eta_p) b_{-p}^* +d_p \\
T_\eta^* b^*_p T_\eta &= \cosh (\eta_p ) b^*_p + \sinh (\eta_p) b_{-p} +d^*_p \end{split} \end{equation}
where the operators $d_p$ are small on states with few excitations, in the sense that  
\begin{equation}\label{eq:dp-bd} \begin{split} \| d_p \xi \| &\leq \frac{C}{N} \left[ |\eta_p| \| (\cN_+ + 1)^{3/2} \xi \| + \| b_p (\cN_+ + 1) \xi \| \right]  
 \end{split} \end{equation}
These bounds are consistent with the observation that $b_p \simeq a_p, b_p^* \simeq a_p^*$ on states exhibiting Bose-Einstein condensation. With (\ref{eq:act-T}) and (\ref{eq:dp-bd}) and with the observation that, by (\ref{eq:etap}), $\| \eta \| \leq C$, uniformly in $N$ ($\| \eta \|$ denotes the $\ell^2$-norm), one can show (see  \cite[Lemma 2.1]{BBCS}) that conjugation with $T_\eta$ leaves the number of excitations essentially unchanged; for every $k \in \bN$, there exists a constant $C_k > 0$ such that 
\begin{equation}\label{eq:TNT} T_\eta^* \cN^k_+ T_\eta \leq C_k (\cN^k_+ + 1) \end{equation}

Conjugating $H_N$ with $U_N$ and then with $T_\eta$, we obtain the renormalized excitation Hamiltonian $\cG_N = T_\eta^* U_N H_N U_N^* T_\eta$, defined on (a dense subspace of) the excitation space $\cF_+^{\leq N}$. In \cite[Prop. 3.2]{BBCS}, it is shown that we can decompose  $\cG_N = C_N + \cQ_N + \cC_N + \cV_N + \cE_N$ where $C_N$ is a constant ($4\pi a_0 N$, up to corrections of order one), $\cQ_N$ is quadratic in creation and annihilation operators, $\cC_N$ is cubic, $\cV_N$ is quartic (and positive) and $\cE_N$ is an error term, negligible on low-energy states. The presence of the cubic term $\cC_N$ is characteristic for the Gross-Pitaevskii regime. It implies that, to get a norm approximation of the ground state vector $\psi_N$, we need to apply an additional cubic transformation. We define 
\begin{equation}\label{eq:Adef} \begin{split} 
A &= \frac1{\sqrt{N}} \sum_{\substack{ r\in P_H, v\in P_L }} \eta_r \big[ \sinh (\eta_v) b^*_{r+v} b^*_{-r}b^*_{-v} +  \cosh (\eta_v) b^*_{r+v} b^*_{-r}b_{v} - \text{h.c.} \big] 
        \end{split}
\end{equation}
where $P_L = \{\, p \in \Lambda_+^* : |p|\leq N^{1/2} \}$ corresponds to low momenta and $P_H =\Lambda_+^* \setminus P_L$ to high momenta (by definition $r + v \not = 0$) 
and we consider the unitary operator $S = e^{-A}$. Similarly to (\ref{eq:TNT}), also the action of $S$ does not change the number of excitations substantially. From \cite[Prop. 4.2]{BBCS} it follows that, for every $k \in \bN$, there exists namely a constant $C_k >0$ such that 
\begin{equation}\label{eq:SNS} e^{\kappa A} \cN^k_+ e^{-\kappa A} \leq C_k (\cN^k_+ + 1)^k \end{equation}
for all $\kappa \in [-1;1]$. 

Conjugating the renormalized excitation Hamiltonian $\cG_N$ with the cubic transformation $S$, we define $\cJ_N = S^*  \cG_N S = S^* T_\eta^* U_N H_N U_N^* T_\eta S$. As shown in \cite[Prop. 3.3]{BBCS}, we can now write $\cJ_N = \wt{C}_N + \wt{\cQ}_N + \cV_N + 
\wt{\cE}_N$ where $\wt{C}_N$ is a constant, $\wt{\cE}_N$ is an error term negligible on low-energy states, and where the quadratic operator $\wt{Q}_N$ is given by 
\[ \cQ_N  = \sum_{p\in \Lambda^*_+} F_p a_p^* a_p + \frac{1}{2} \sum_{p \in \Lambda^*_+} G_p \left[ a_p^* a_{-p}^* + a_p a_{-p} \right] \]
with 
\[ \begin{split} 
F_p &= p^2 [ \sigma^2_p + \gamma^2_p] + \left( \widehat{V} (./N) * \widehat{f}_{N,\ell} \right)_p (\sigma_p + \gamma_p)^2  \\
G_p &= 2 p^2 \sigma_p \gamma_p + \left( \widehat{V} (./N) * \widehat{f}_{N,\ell} \right)_p (\sigma_p + \gamma_p)^2   \end{split} \]
To diagonalize $\wt{Q}_N$, we apply a final generalized Bogoliubov transformation 
\[ T_\tau  = \exp \left[ \frac{1}{2} \sum_{p \in \Lambda^*_+} \tau_p \left( b_p^* b_{-p}^* - b_p b_{-p} \right) \right] \]
with coefficients $\tau_p$ defined through 
\begin{equation}\label{eq:taup-def}  \tanh (2\tau_p) = - \frac{G_p}{F_p}  \end{equation}
As shown in \cite[Lemma 5.1]{BBCS}, we have $| \tau_p | \leq C |p|^{-4}$. This implies that $\| \tau \| \leq C$ uniformly in $N$ and thus, similarly to (\ref{eq:TNT}), that for every $k \in \bN$ there exists $C_k > 0$ such that 
\begin{equation}\label{eq:TtNTt} T_\tau^* \cN_+^k T_\tau \leq C_k (\cN^k_+ + 1) \end{equation}
It follows that $\cM_N = T_\tau^* \cJ_N T_\tau = T_\tau^*  S^* T_\eta^* U_N H_N U_N^* T_\eta S T_\tau$ is essentially quadratic and diagonal (up to the positive quartic term $\cV_N$ and an error that is negligible on low-energy states). This implies that the vacuum vector $\Omega = \{ 1, 0, \dots , 0\} \in \cF_+^{\leq N}$ approximate the ground state vector of $\cM_N$ in norm. In fact, it is shown in \cite[Sect. 6]{BBCS} that there is a phase $\omega \in [0;2\pi]$ such that  
\begin{equation}\label{eq:psi-app} \left\| \psi_N - e^{i\omega} U_N^* T_\eta S T_\tau  \Omega \right\| \leq C N^{-1/4} \end{equation}

\section{Proof of Theorem \ref{thm:main}}
\label{sec:proof}

The proof of the main theorem is based on (\ref{eq:psi-app}); after replacing $\psi_N$ with its approximation $T_\tau S T_\eta U_N^* \Omega$, we can control the action of the unitary operator $U_N$ and of the generalized Bogoliubov transformations $T_\eta$ and $T_\tau$ using (\ref{eq:rules}) and, respectively, (\ref{eq:act-T}), (\ref{eq:dp-bd}). To control the action of the cubic phase $S = e^{-A}$, on the other hand, we will make use of the following lemma.

\begin{lemma}
\label{lemma:commA}
Let $A$ be as defined in (\ref{eq:Adef}). There is a constant $C > 0$ such that 
\begin{align*}
| \langle \xi_1, \left[ b(h) , A \right] \xi_2 \rangle |
\leq&   \frac{C \| h \|}{\sqrt{N}} \,  \| ( \mathcal{N}_+ +1)^{1/2} \xi_1 \| \; \| ( \mathcal{N}_+ +1)^{1/2} \xi_2 \| , \\
| \langle \xi_1 ,  \left[ b^*(h), A \right]  \xi_2 \rangle | \leq&  \frac{C \| h \|}{\sqrt{N}}  \| ( \mathcal{N}_+ +1)^{1/2} \xi_1 \| \; \| ( \mathcal{N}_+ +1)^{1/2} \xi_2 \| 
\end{align*}
for every $\xi_1, \xi_2 \in \mathcal{F}_{+}^{\leq N}$ and every $h \in \ell^2 (\Lambda_+^*)$ (the notation $\| h \|$ indicates the $\ell^2$-norm of $h$). The modified creation and annihilation operators $b(h), b^* (h)$ are defined as in (\ref{eq:bb}). \end{lemma}
\begin{proof}
For $p \in \Lambda^*_+$ let $\gamma_v = \cosh \eta_v$ and $\sigma_v = \sinh \eta_v$. We use the definition (\ref{eq:Adef}) and the commutation relations (\ref{eq:ccr-mod}) to compute 
\begin{align*}
\left[ b_p, A \right]  =& \frac{1}{\sqrt{N}} \sum_{r \in P_H, v \in P_L} \eta_r \sigma_v \left[  \left( 1-\frac{\mathcal{N}_+}{N}\right) \delta_{p,r+v} b_{-r}^* b_{-v}^* -\frac{1}{N} a^*_{r+v}a_p b_{-r}^* b_{-v}^* \right]  \\
&+ \frac{1}{\sqrt{N}} \sum_{r \in P_H, v \in P_L} \eta_r \sigma_v   \left[  b_{r+v}^* \left( 1-\frac{\mathcal{N}_+}{N}\right) \delta_{p,-r} b_{-v}^* -\frac{1}{N}b_{r+v}^* a^*_{-r}a_p  b_{-v}^* \right]  \\
&+ \frac{1}{\sqrt{N}} \sum_{r \in P_H, v \in P_L} \eta_r \sigma_v \left[  b_{r+v}^* b_{-r}^*\left( 1-\frac{\mathcal{N}_+}{N}\right) \delta_{p,-v}  -\frac{1}{N}b_{r+v}^*b_{-r}^* a^*_{-v}a_p   \right]  \\
&+  \frac{1}{\sqrt{N}} \sum_{r \in P_H, v \in P_L} \eta_r \gamma_v \left[  \left( 1-\frac{\mathcal{N}_+}{N}\right) \delta_{p,r+v} b_{-r}^* b_{v} -\frac{1}{N} a^*_{r+v}a_p b_{-r}^* b_{v} \right]   \\
&+ \frac{1}{\sqrt{N}} \sum_{r \in P_H, v \in P_L} \eta_r \gamma_v \left[  b_{r+v}^* \left( 1-\frac{\mathcal{N}_+}{N}\right) \delta_{p,-r} b_{v} -\frac{1}{N}b_{r+v}^* a^*_{-r}a_p  b_{v} \right]  \\ 
&- \frac{1}{\sqrt{N}} \sum_{r \in P_H, v \in P_L} \eta_r \gamma_v  \left[  \left( 1-\frac{\mathcal{N}_+}{N}\right) \delta_{p,v} b_{-r} b_{r+v} -\frac{1}{N} a^*_{v}a_p b_{-r} b_{r+v} \right]  \\ 
=& \frac{1}{\sqrt{N}} \sum_{r \in P_H, v \in P_L} \eta_r \left[  \left(  1- \frac{\mathcal{N}_+}{N} \right)b^*_{-r} \left( \sigma_v b^*_{-v} + \gamma_v b_v \right)  \delta_{p,r+v}  \right. 
\\ &\hspace{3cm} + \left( 1- \frac{\mathcal{N}_+ -1}{N} \right) b_{r+v}^* ( \sigma_v b^*_{-v}  + \gamma_v b_v) \delta_{p,-r} \\
& \hspace{3cm} + \left. \sigma_v b_{r+v}^*b_{-r}^* \left( 1- \frac{\mathcal{N}_+}{N} \right) \delta_{p,-v} - \gamma_v \left( 1- \frac{\mathcal{N}_+}{N} \right) b_{-r} b_{r+v}  \delta_{p,v} \right] \\
&-\frac{1}{N^{3/2}} \sum_{r \in P_H, v \in P_L} \eta_r \left[ ( a_{r+v}^*a_p b^*_{-r} + b^*_{r+v}a^*_{-r}a_p ) (\sigma_v b_{-v}^* + \gamma_v b_v )\right. \\
& \hspace{4cm}\left. + \sigma_v b^*_{r+v}b^*_{-r}a^*_{-v}a_p -\gamma_v a_v^*a_p b_{-r}b_{r+v} \right].
\end{align*}
Thus, we write
\begin{align*}
\left[ b(h), A \right] = \sum_{p \in \Lambda_+^*} \overline{h} (p) \left[ b_p, A \right] = \sum_{j=1}^5 D_j 
\end{align*}
where
\begin{align*}
D_1 :=& \frac{1}{\sqrt{N}} \sum_{r \in P_H, v \in P_L} \overline{h}( r+v) \;  \eta_r  \left(  1- \frac{\mathcal{N}_+}{N} \right)b^*_{-r} \left( \sigma_v b^*_{-v} + \gamma_v b_v \right)  \\
D_2:=&\frac{1}{\sqrt{N}}  \sum_{r \in P_H, v \in P_L} \overline{h}( -r) \; \eta_r \left( 1- \frac{\mathcal{N}_+ -1}{N} \right) b_{r+v}^* ( \sigma_v b^*_{-v}  + \gamma_v b_v) \\
D_3:= & \frac{1}{\sqrt{N}} \sum_{r \in P_H, v \in P_L}\left( \overline{ h}( -v) \; \eta_r \sigma_v b_{r+v}^*b_{-r}^* \left( 1- \frac{\mathcal{N}_+}{N} \right) + \overline{h}( v) \eta_r \gamma_v  \left( 1- \frac{\mathcal{N}_+}{N} \right) b_{-r} b_{r+v} \right)  \\
D_4 :=&- \frac{1}{N^{3/2}}  \sum_{p \in \Lambda_+^*}  \sum_{r \in P_H, v \in P_L}\overline{h}(p) \; \eta_r ( a_{r+v}^*a_p b^*_{-r} + b^*_{r+v}a^*_{-r}a_p ) (\sigma_v b_{-v}^* + \gamma_v b_v ) \\
D_5 :=&  \frac{1}{N^{3/2}}  \sum_{p \in \Lambda_+^*}   \sum_{r \in P_H, v \in P_L}\overline{h}(p) \; \eta_r \left(  \sigma_v b^*_{r+v}b^*_{-r}a^*_{-v}a_p -\gamma_v a_v^*a_p b_{-r}b_{r+v}  \right). 
\end{align*}
In the following we will write, for $j \in \{ 1, 2 , \dots , 5 \}$, $D_j = D_{j,\sigma} + D_{j, \gamma}$, with $D_{j,\sigma}$, $D_{j,\gamma}$ denoting the terms in $D_j$ containing the factor $\sigma_v$ and, respectively, the factor $\gamma_v$. Since, by (\ref{eq:etap}), $| \eta_r | \leq Cr^{-2} , | \sigma_v | \leq Cv^{-2}$, we obtain 
\begin{align*}
\vert\langle \xi_1 ,  D_{1,\sigma} \xi_2 \rangle \vert \leq& \frac{C}{\sqrt{N}} \| (\mathcal{N}_++1)^{1/2} \xi_2\| \sum_{r \in P_H, v \in P_L} |h( r+v)| \;  |\eta_r  | \; | \sigma_v| \;  \| b_{-v}  b_{-r} \left(  1- \frac{\mathcal{N}_+}{N} \right)\left(\mathcal{N}_++1\right)^{-1/2}\xi_1 \| \\
\leq& \frac{C \| h \|_{\infty}}{\sqrt{N}} \| ( \mathcal{N}_+ +1)^{1/2}\xi_2 \| \sum_{r \in P_H, v \in P_L} | \eta_r | \; | \sigma_v| \; \| b_{-v}  b_{-r} \left(\mathcal{N}_+  + 1\right)^{-1/2} \xi_1 \|\\
\leq& \frac{C \| h \|_{\infty}}{\sqrt{N}}  \| ( \mathcal{N}_+ +1)^{1/2}\xi_2 \|\left( \sum_{r \in P_H, v \in P_L}  | \sigma_v|^2 | \eta_r |^2 \right)^{1/2} \\ &\hspace{3cm} \times \left( \sum_{r \in P_H, v \in P_L} \| b_{-v}  b_{-r} \; \left( \mathcal{N}_+ +1\right)^{-1/2}\xi_1 \|^2 \right)^{1/2} \\
\leq&\frac{C \| h \|_\infty}{\sqrt{N}} \| ( \mathcal{N}_+ +1)^{1/2}\xi_1 \| \; \|  \left(\mathcal{N}_++1\right)^{1/2}\xi_2 \|.
\end{align*}
For the term $D_{1,\gamma}$, we estimate 
\begin{align*}
\vert \langle \xi_1 , D_{1,\gamma} \xi_2 \rangle \vert  \leq& \frac{1}{\sqrt{N}} \sum_{r \in P_H, v \in P_L} |h(r+v)|\;  |\eta_r| \; | \gamma_v| \| b_{v} \xi_1 \| \; \|   b_{-r} \left( 1- \frac{\mathcal{N}_+}{N} \right) \xi_2 \|  \\
\leq& \frac{C \| h \|}{\sqrt{N}} \| \mathcal{N}_+^{1/2} \xi_1 \| \; \sum_{r \in P_H}  | \eta_r| \; \| b_{-r} \xi_2 \| \\
\leq&  \frac{C \| h \|}{\sqrt{N}} \| \mathcal{N}_+^{1/2} \xi_1 \| \; \| \mathcal{N}_+^{1/2} \xi_2 \|,
\end{align*}
where we used that $  | \eta_r | \leq Cr^{-2}, \; | \gamma_v | \leq C $.

The term $D_2$ can be bounded in the same way. For the third term $D_3$ we find (using the same arguments) on the one hand
\begin{align*}
\vert \langle \xi_1, D_{3, \sigma} \xi_2 \rangle \vert \leq&  \frac{C}{\sqrt{N}} \| ( \mathcal{N}_+ +1)^{1/2} \xi_2 \| \sum_{r \in P_H, v \in P_L} | h(-v)| \; | \eta_r| \; | \sigma_v| \;  \| b_{-r} b_{r+v} (\mathcal{N}_++1)^{-1/2}\xi_1 \|  \\
 \leq& \frac{C \| h \|_\infty}{\sqrt{N}} \| ( \mathcal{N}_+ +1)^{1/2}\xi_1 \| \; \|( \mathcal{N}_+ +1)^{1/2} \xi_2 \| ,
\end{align*}
and on the other hand
\begin{align*}
\vert \langle \xi_1 , D_{3, \gamma} \xi_2 \rangle \vert \leq& \frac{C}{\sqrt{N}} \| ( \mathcal{N}_+ + 1)^{1/2}\xi_1 \| \sum_{r \in P_H, v \in P_L} | h(v)| \; | \eta_r| \; | \gamma_v| \; \| b_{-r} b_{r+v} (\mathcal{N}_+ + 1)^{-1/2}\xi_2 \| \\ \leq& \frac{C \| h \|}{\sqrt{N}} \|( \mathcal{N}_++1)^{1/2} \xi_	1 \| \; \|  \mathcal{N}_+^{1/2} \xi_2 \|. 
\end{align*}
To control $D_4$, we rearrange creation and annihilation operators in normal order. 
We split
\begin{align*}
D_{4, \sigma} 
=& - \frac{2}{N^{3/2}}  \sum_{p \in \Lambda_+^*}  \sum_{r \in P_H, v \in P_L}\overline{h}(p) \; \eta_r \; \sigma_v \; b_{-v}^* b^*_{-r} a_{r+v}^*a_p  \\
&- \frac{2}{N^{3/2}}   \sum_{r \in P_H, v \in P_L}\overline{h}(-v) \; \eta_r \; \sigma_v   b^*_{-r} b_{r+v}^* \\
&- \frac{1}{N^{3/2}}  \sum_{r \in P_H, v \in P_L}\overline{h}(-r) \; \eta_r \; \sigma_v   b^*_{r+v} b_{-v}^*  = \sum_{i= 1}^3 D_{4, \sigma}^{(i)}.
\end{align*}
The two terms $D_{4, \sigma}^{(2)}$ and $D_{4, \sigma}^{(3)}$ can be bounded like the previous terms. We find 
\begin{align*}
\vert \langle \xi_1 , D_{4, \sigma}^{(j)} \xi_2 \rangle \vert \leq \frac{C \| h \|_\infty }{N^{3/2}} \| ( \mathcal{N}_+ +1)^{1/2} \xi_1 \| \;  \| ( \mathcal{N}_+ +1)^{1/2} \xi_2 \|
\end{align*}
for $j=2,3$. As for the term $D_{4,\sigma}^{(1)}$, we estimate  
\begin{align*}
\vert \langle \xi_1, D_{4, \sigma}^{(1)}\xi_2 \rangle \vert \leq& \frac{2}{N^{3/2}}  \sum_{p \in \Lambda_+^*}  \sum_{r \in P_H, v \in P_L} |h(p)| \; |\eta_r| \; |\sigma_v| \; \| b_{-v} b_{-r} \xi_1 \| \; \| a_p ( \mathcal{N}_++1)^{1/2} \xi_2 \| \\
\leq& \frac{C}{N^{3/2}} \left( \sum_{p \in \Lambda_+^*}\sum_{r \in P_H, v \in P_L} | h(p) |^2 \; \| b_{-v} b_{-r} \xi_1 \|^2 \right)^{1/2} \\ &\hspace{2cm} \times \left( \sum_{p \in \Lambda_+^*} \sum_{r \in P_H, v \in P_L} | \eta_r |^2 \; | \sigma_v|^2 \| ( \mathcal{N}_+ +1) a_p \xi_2 \|^2 \right)^{1/2}\\
\leq& \frac{C \| h \|}{N^{3/2}} \|( \mathcal{N}_++1) \xi_1 \| \; \| ( \mathcal{N}_++1) \xi_2 \| \leq \frac{C \| h \|}{N^{1/2}} \| ( \mathcal{N}_+ +1)^{1/2} \xi_1 \| \; \| ( \mathcal{N}_+ +1)^{1/2} \xi_2 \| .
\end{align*}
Analogously, we split 
\begin{equation}\label{eq:D4g} 
\begin{split} 
D_{4, \gamma} :=&- \frac{1}{N^{3/2}}  \sum_{p \in \Lambda_+^*}  \sum_{r \in P_H, v \in P_L}\overline{h}(p) \; \eta_r \; \gamma_v \; ( b^*_{-r} a_{r+v}^*a_p + b^*_{r+v}a^*_{-r}a_p )  b_v \\
&-  \frac{1}{N^{3/2}}  \sum_{r \in P_H, v \in P_L}\overline{h}(-r) \; \eta_r \; \gamma_v \;  b_{r+v}^* b_v = \sum_{i=1}^2 D_{4, \gamma}^{(i)}.
\end{split}
\end{equation}
For the second term, we find
\begin{align*}
\vert \langle \xi_1, D_{4,\gamma}^{(2)} \xi_2 \rangle \vert \leq& \frac{1}{N^{3/2}}  \sum_{r \in P_H, v \in P_L} |h (-r)| \; |\eta_r| \; |\gamma_v |\; \| b_{r+v} \xi_1\| \; \| b_v \xi_2 \|  \\ \leq& \frac{C \| h \| }{N^{3/2}} \| ( \mathcal{N}_+ +1)^{1/2} \xi_1 \| \; \| ( \mathcal{N}_+ +1)^{1/2} \xi_2 \|, 
\end{align*}
Similarly, the first term on the r.h.s. of (\ref{eq:D4g}) can be bounded by 
\begin{align*}
\vert \langle \xi_1, D_{4,\gamma}^{(1)} \xi_2 \rangle \vert \leq& \frac{2}{N^{3/2}}  \sum_{p \in \Lambda_+^*}  \sum_{r \in P_H, v \in P_L}| \overline{h}(p)| \; |\eta_r| \; |\gamma_v | \;   \| a_{r+v}b_{-r} \xi_1 \| \; \| a_p b_v \xi_2 \|  \\
\leq& \frac{C}{N^{3/2}} \left( \sum_{p \in \Lambda_+^*}\sum_{r \in P_H, v \in P_L} | h(p)|^2 \; \| a_{r+v} b_{-r} \xi_1 \|^2 \right)^{1/2} \left( \sum_{p \in \Lambda_+^*} \sum_{r, \in P_H, v \in P_L} | h( p)|^2 \| a_{p} b_v \xi_2 \| \right)^{1/2} \\
\leq& \frac{C \| h \|}{N^{3/2}} \| ( \mathcal{N}_+ +1) \xi_1 \| \; \| ( \mathcal{N}_+ +1) \xi_2 \| \leq \frac{C \| h \|}{N^{1/2}} \| ( \mathcal{N}_++1)^{1/2} \xi_1 \| \; \| ( \mathcal{N}_+ +1)^{1/2} \xi_2 \| .
\end{align*}
The term $D_5$ can be estimated like the normally ordered contributions in $D_4$. We conclude that 
\begin{align*}
| \langle \xi_1, \left[ b(h), A \right] \xi_2 \rangle |
\leq&  \frac{C \| h \|}{\sqrt{N}}  \| ( \mathcal{N}_+ +1)^{1/2} \xi_1 \| \; \| ( \mathcal{N}_+ +1)^{1/2} \xi_2 \| . 
\end{align*}
Since $\left[b^*(h), A \right] = \left[b(h) , A \right]^*$ (because $A$ is skew-symmetric), the second inequality in Lemma \ref{lemma:commA} follows from the first one. 
\end{proof} 

The next lemma will also be useful to pass powers of the number of particles operator $\cN_+$ through modified Weyl operators (whose action is not explicit). 
\begin{lemma}\label{lm:weyl}
For every $j \in \bN$ there exists a constant $C > 0$ such that 
\[ \langle \xi, e^{-i \phi (h)} (\cN_+ + \alpha)^j e^{i \phi (h)} \xi \rangle \leq C \langle \xi, (\cN_+ + \alpha  + \| h \|^2)^j \xi \rangle \]
for all $\alpha \geq 1$, $\xi \in \cF_+^{\leq N}$, $h \in \ell^2 (\Lambda^*_+)$. Here we used the notation $\phi (h) = b(h) + b^* (h)$, with the modified creation and annihilation operators defined in (\ref{eq:bb}). 
\end{lemma}

\begin{proof}
For a fixed $\xi \in \cF_+^{\leq N}$ with $\| \xi \| =1$, we define the function $f: \mathbb{R} \rightarrow \mathbb{R}$ through
\begin{align*}
f (t) = \langle \xi,  e^{-it \phi (h)} (\cN_+ + \alpha)^j e^{i t \phi (h)} \xi \rangle  
\end{align*}
We compute 
\begin{align*}
i f' (t) =& \langle \xi, e^{-it \phi (h)} \left[ ( \mathcal{N}_+ + \alpha)^j, \phi (h)  \right] e^{i t \phi (h)} \xi \rangle \\
=& \sum_{\ell=0}^{j-1} \langle \xi, e^{-it\phi (h)} ( \mathcal{N}_+ + \alpha)^{\ell} \left[ \mathcal{N}_+  , \phi (h) \right] ( \mathcal{N} +\alpha)^{j-1-\ell} e^{it \phi (h)} \xi \rangle.
\end{align*}
The commutation relations (\ref{eq:ccr-mod}) imply that
\begin{align*}
i f' (t) =&  - \sum_{\ell=0}^{j-1} \langle \xi, e^{-it\phi (h)} ( \mathcal{N}_+ + \alpha)^{\ell} \; b(h) ( \mathcal{N}_+ + \alpha)^{j-1-\ell}  e^{it \phi (h)} \xi \rangle \\
& + \sum_{\ell=0}^{j-1} \langle \xi, e^{-it \phi (h)} ( \mathcal{N}_+ + \alpha)^{\ell}\; b^*(h) ( \mathcal{N}_+ + \alpha)^{j-1-\ell} e^{it \phi (h)}  \xi \rangle.
\end{align*}
With $b(h) \cN_+ = (\cN_+ + 1) b(h)$ and $b^* (h) \cN_+ = (\cN_+ - 1) b^* (h)$, we find a constant $C$ (depending only on $j$) such that 
\[ |f' (t)| \leq C \| h \| \langle e^{it \phi (h)} \xi , (\cN_+ + \alpha)^{j-1/2} e^{it \phi (h)} \xi \rangle \leq C \| h \| f(t)^{1-\frac{1}{2j}} \]
where we used the condition $\alpha > 1$ and, in the last inequality, the normalization $\| \xi \|=1$. Gronwall's lemma implies that 
\[ |f(t)| \leq C (f(0) +\| h \|^{2j} t^{2j}) \]
Thus
\[ \langle e^{i \phi (h)} \xi, (\cN_+ + \alpha)^j e^{i \phi (h)} \xi \rangle \leq C \langle \xi, (\cN_+ +\alpha + \| h \|^2)^j \xi \rangle \,. \]
\end{proof}

We are now ready to prove our main result.
\begin{proof}[Proof of Theorem \ref{thm:main}]
We write
\begin{align*}
\mathbb{E}_{\psi_N}\left[ g_1(\mathcal{O}_{1}) \dots g_k ( \mathcal{O}_{k}) \right] &= \langle \psi_{N} , g_1( \mathcal{O}_{1} ) \dots g_k ( \mathcal{O}_{k} ) \psi_{N} \rangle \\
&= \int ds_1 \dots ds_k \; \widehat{g}_1(s_1) \dots \widehat{g}_k( s_k) \;  \langle \psi_{N}, e^{i s_1 \mathcal{O}_{1}} \dots e^{i s_k \mathcal{O}_{k}}\psi_{N} \rangle 
\end{align*}
With (\ref{eq:psi-app}), we obtain that 
\begin{equation}\label{eq:pro1} \begin{split} 
&\left| \mathbb{E}_{\psi_N}\left[ g_1(\mathcal{O}_{1}) \dots g_k ( \mathcal{O}_{k}) \right] - \int ds_1 \dots ds_k \; \widehat{g}_1(s_1) \dots \widehat{g}_k( s_k) \;   \langle U_N^* T_\eta S T_\tau \Omega , e^{i s_1 \mathcal{O}_{1}} \dots e^{i s_k \mathcal{O}_{k}}  U_N^* T_\eta S T_\tau \Omega \rangle \right| \\ &\hspace{12cm} \leq C N^{-1/4} \prod_{j=1}^k \| \widehat{g}_j \|_1  \end{split}  \end{equation}
We split the computation of the expectation $\langle U_N^* T_\eta S T_\tau \Omega , e^{i s_1 \mathcal{O}_{1}} \dots e^{i s_k \mathcal{O}_{k}}  U_N^* T_\eta S T_\tau \Omega \rangle$ in several steps. 

\medskip

{\it Step 1. Action of $U_N$.} We set $\xi_1 = T_\eta S T_\tau \Omega \in \cF_+^{\leq N}$. Then we have 
\begin{equation}\label{eq:step1} \begin{split} 
&\left| \langle U_N^* \xi_1, e^{i s_1 \mathcal{O}_{1}} \dots e^{i s_k \mathcal{O}_{k}}  U_N^* \xi_1 \rangle  - \langle \xi_1, e^{is_1 \phi (q_0 O_1 \ph_0)} \dots e^{is_k \phi (q_0 O_k \ph_0)} \xi_1 \rangle \right| \\ &\hspace{5cm} \leq \frac{C}{\sqrt{N}} \left( \sum_{j=1}^k \| O_j \| |s_j| \right) \left( 1 + \sum_{j=1}^k \| q_0 O_j \ph_0 \|^2 |s_j|^2 \right) \end{split} \end{equation}
Here we use the notation $\phi (h) = b(h) + b^* (h)$, with the modified creation and annihilation operators introduced in (\ref{eq:bb}).

\medskip

With $\wt{O}_j = O_j - \langle \ph_0, O_j \ph_0\rangle$ and denoting $d\Gamma (A) = \sum_{\ell=1}^N A_\ell$, we can decompose  
\[ \cO_j = \frac{1}{\sqrt{N}} \left[ d\Gamma (q_0 \wt{O}_j q_0) + d\Gamma (q_0 O_j p_0) + d\Gamma (p_0 O_j q_0) \right] .\]
With (\ref{eq:rules}), we obtain 
\[ U_N^* \cO_j U_N = \frac{1}{\sqrt{N}} d\Gamma (q_0 \wt{O}_j q_0) + \phi (q_0 O_j \ph_0) \]
Hence, we find 
\[ \langle U_N^* \xi_1, e^{i s_1 \mathcal{O}_{1}} \dots e^{i s_k \mathcal{O}_{k}}  U_N^* \xi_1 \rangle = \langle  \xi_1, \prod_{j=1}^k e^{i s_j \left[ \frac{1}{\sqrt{N}} d\Gamma (q_0 \wt{O}_j q_0) + \phi (q_0 O_j \ph_0) \right]} \xi_1 \rangle \]
and 
\[ \begin{split}  \langle U_N^* \xi_1, &e^{i s_1 \mathcal{O}_{1}} \dots e^{i s_k \mathcal{O}_{k}}  U_N^* \xi_1 \rangle - \langle \xi_1, e^{is_1 \phi (q_0 O_1 \ph_0)} \dots e^{is_k \phi (q_0 O_k \ph_0)} \xi_1 \rangle \\ &= \sum_{m=1}^k \langle \xi_1, \prod_{j=1}^{m-1} e^{i s_j \left[ \frac{1}{\sqrt{N}} d\Gamma (q_0 \wt{O}_j q_0) + \phi (q_0 O_j \ph_0) \right]} \\ &\hspace{2cm} \times \left( e^{i s_m \left[ \frac{1}{\sqrt{N}} d\Gamma (q_0 \wt{O}_m q_0) + \phi (q_0 O_m \ph_0) \right]} -  e^{i s_m \phi (q_0 O_m \ph_0)} \right)  \prod_{j=m+1}^{k} e^{i s_j  \phi (q_0 O_j \ph_0)} \xi_1 \rangle 
\end{split} \] 
Writing the difference in the parenthesis as an integral, we arrive at
\[ \begin{split}  \langle U_N^* \xi_1, &e^{i s_1 \mathcal{O}_{1}} \dots e^{i s_k \mathcal{O}_{k}}  U_N^* \xi_1 \rangle - \langle \xi_1, e^{is_1 \phi (q_0 O_1 \ph_0)} \dots e^{is_k \phi (q_0 O_k \ph_0)} \xi_1 \rangle \\ &= \frac{i}{\sqrt{N}} \sum_{m=1}^k \int_0^{s_m} dt \, \langle \xi_1, \prod_{j=1}^{m-1} e^{i s_j \left[ \frac{1}{\sqrt{N}} d\Gamma (q_0 \wt{O}_j q_0) + \phi (q_0 O_j \ph_0) \right]} \\ &\hspace{1cm} \times 
e^{i (s_m -t) \left[ \frac{1}{\sqrt{N}} d\Gamma (q_0 \wt{O}_m q_0) + \phi (q_0 O_m \ph_0) \right]}  d\Gamma (q_0 \wt{O}_j q_0) e^{i t \phi (q_0 O_m \ph_0)} \prod_{j=m+1}^{k} e^{i s_j  \phi (q_0 O_j \ph_0)} \xi_1 \rangle 
\end{split} \] 
Since $\| d\Gamma (A) \xi \| \leq \| A \| \| \cN_+ \xi \|$, we conclude that 
\[ \begin{split}  &\left| \langle U_N^* \xi_1, e^{i s_1 \mathcal{O}_{1}} \dots e^{i s_k \mathcal{O}_{k}}  U_N^* \xi_1 \rangle - \langle \xi_1, e^{is_1 \phi (q_0 O_1 \ph_0)} \dots e^{is_k \phi (q_0 O_k \ph_0)} \xi_1 \rangle \right| \\ & \hspace{1cm} \leq \frac{1}{\sqrt{N}} \sum_{m=1}^k \| q_0 \wt{O}_m q_0 \| \int_0^{s_m} dt \,  \Big\| \cN_+ e^{i t \phi (q_0 O_m \ph_0)} \prod_{j=m+1}^{k} e^{i s_j  \phi (q_0 O_j \ph_0)} \xi_1 \Big\| 
\end{split} \] 
With Lemma \ref{lm:weyl}, we find
\[ \begin{split}  &\left| \langle U_N^* \xi_1, e^{i s_1 \mathcal{O}_{1}} \dots e^{i s_k \mathcal{O}_{k}}  U_N^* \xi_1 \rangle - \langle \xi_1, e^{is_1 \phi (q_0 O_1 \ph_0)} \dots e^{is_k \phi (q_0 O_k \ph_0)} \xi_1 \rangle \right| \\ & \hspace{4cm} \leq \frac{C}{\sqrt{N}} \sum_{m=1}^k |s_m| \| q_0 \wt{O}_m q_0 \|  \,  \Big\| \Big(\cN_+  + 1 + \sum_{j=m}^k s_j^2 \| q_0 O_j \ph_0 \|^2 \Big) \xi_1 \Big\| 
\end{split} \] 
Using (\ref{eq:TNT}), (\ref{eq:SNS}) and (\ref{eq:TtNTt}), we conclude that $\langle \xi_1, \cN_+^2 \xi_1 \rangle \leq C$, uniformly in $N$. Therefore 
\[ \begin{split}  &\left| \langle U_N^* \xi_1, e^{i s_1 \mathcal{O}_{1}} \dots e^{i s_k \mathcal{O}_{k}}  U_N^* \xi_1 \rangle - \langle \xi_1, e^{is_1 \phi (q_0 O_1 \ph_0)} \dots e^{is_k \phi (q_0 O_k \ph_0)} \xi_1 \rangle \right| \\ & \hspace{4cm} \leq \frac{C}{\sqrt{N}} \sum_{m=1}^k |s_m| \| q_0 \wt{O}_m q_0 \|  \,  \Big( 1 + \sum_{j=m}^k s_j^2 \| q_0 O_j \ph_0 \|^2 \Big)  
\end{split} \] 
which implies (\ref{eq:step1}). 

\medskip

{\it Step 2. Action of $T_\eta$.} Let $\xi_2 = ST_\tau \Omega$, so that $\xi_1 = T_\eta \xi_2$. Then we have 
\begin{equation}\label{eq:step2} \begin{split} &\left| \langle \xi_1, e^{is_1 \phi (q_0 O_1 \ph_0)} \dots e^{is_k \phi (q_0 O_k \ph_0)} \xi_1 \rangle - \langle \xi_2, e^{is_1 \phi (h_1)} \dots e^{is_k \phi (h_k)} \xi_2 \rangle \right| \\ & \hspace{6cm} \leq \frac{C}{N} \left( \sum_{m=1}^k |s_m| \| q_0 O_j \ph_0 \| \right) \left( 1 + \sum_{j=1}^k s_j^2 \| q_0 O_j \ph_0 \|^2 \right)^{3/2} \end{split} \end{equation}
where $h_j (p) =  \widehat{(q_0 O_j \ph_0)} (p) \cosh (\eta_p) + 
\widehat{(\overline{q_0 O_j \ph_0})} (p) \sinh (\eta_p)$, for all $j=1,\dots, k$. 

\medskip

To prove (\ref{eq:step2}), we use (\ref{eq:act-T}) to write 
\begin{equation}\label{eq:def-h} \begin{split} T_\eta^* &\phi (q_0 O_j \ph_0) T_\eta \\ = \; &\sum_{p \in \Lambda^*_+} \widehat{(q_0 O_j \ph_0)} (p) \left[ \cosh (\eta_p)  b_p^* + \sinh (\eta_p) b_{-p} + d_p^* \right] +  \overline{\widehat{(q_0 O_j \ph_0)} (p)} \left[ \cosh (\eta_p) b_p + \sinh (\eta_p) b_{-p}^* + d_p \right] \\ = \; &\phi (h_j) + D_j \end{split} \end{equation}
where 
\[ D_j = \sum_{p \in \Lambda^*_+} \widehat{(q_0 O_j \ph_0)} (p) d_p^* + \overline{\widehat{(q_0 O_j \ph_0)} (p)} d_p \]
We have
\[ \begin{split} \langle  \xi_1, e^{is_1 \phi (q_0 O_1 \ph_0)} \dots e^{is_k \phi (q_0 O_k \ph_0)} \xi_1 \rangle &=  \langle  \xi_2, T_\eta^* e^{is_1 \phi (q_0 O_1 \ph_0)} \dots e^{is_k \phi (q_0 O_k \ph_0)} T_\eta \xi_2 \rangle
 \\ &=  \langle  \xi_2,  e^{is_1 \left[ \phi (h_1) + D_1\right] } \dots e^{is_k \left[ \phi (h_k) + D_k\right]} \xi_2 \rangle 
 \end{split} \]
and
\[ \begin{split} \langle  \xi_1, &e^{is_1 \phi (q_0 O_1 \ph_0)} \dots e^{is_k \phi (q_0 O_k \ph_0)} \xi_1 \rangle - \langle \xi_2, e^{is_1 \phi (h_1)} \dots e^{is_k \phi (h_k)} \xi_2 \rangle 
\\ &= \sum_{m=1}^k \langle \xi_2, \prod_{j=1}^{m-1} e^{is_j \left[ \phi (h_j) + D_j\right] }  \left( e^{is_m \left[ \phi (h_m) + D_m\right] } - e^{is_m \phi (h_m)} \right) \prod_{j=m+1}^k  e^{is_j \phi (h_j)} \xi_2 \rangle \\
&= i \sum_{m=1}^k \int_0^{s_m} dt \, 
\langle \xi_2, \prod_{j=1}^{m-1} e^{is_j \left[ \phi (h_j) + D_j\right] }  e^{i (s_m -t) \left[ \phi (h_m) + D_m \right] } D_m e^{i t \phi (h_m)} \prod_{j=m+1}^k  e^{is_j \phi (h_j)} \xi_2 \rangle\end{split} \]
It follows that 
\[ \begin{split} &\left| \langle  \xi_1, e^{is_1 \phi (q_0 O_1 \ph_0)} \dots e^{is_k \phi (q_0 O_k \ph_0)} \xi_1 \rangle - \langle \xi_2, e^{is_1 \phi (h_1)} \dots e^{is_k \phi (h_k)} \xi_2 \rangle \right| 
\\ &\hspace{2cm} \leq \sum_{m=1}^k \sum_{p\in \Lambda^*_+} |\widehat{q_0 O_m \ph_0} (p)| 
\int_0^{s_m} dt \, \Big[ \Big\| d_p e^{i t \phi (h_m)} \prod_{j=m+1}^k  e^{is_j \phi (h_j)} \xi_2 \Big\| \\ &\hspace{7cm}+  \Big\| d_p e^{-i (s_m-t) \left[ \phi (h_m) + D_m \right]} \prod_{j=1}^{m-1}  e^{-is_j \left[ \phi (h_j) + D_j \right]} \xi_2 \Big\| \Big]  \end{split} \]
With (\ref{eq:dp-bd}), we obtain 
\[ \begin{split} &\left| \langle  \xi_1, e^{is_1 \phi (q_0 O_1 \ph_0)} \dots e^{is_k \phi (q_0 O_k \ph_0)} \xi_1 \rangle - \langle \xi_2, e^{is_1 \phi (h_1)} \dots e^{is_k \phi (h_k)} \xi_2 \rangle \right| 
\\ & \hspace{1cm} \leq \frac{C}{N} \sum_{m=1}^k  \| q_0 O_m \ph_0 \| \int_0^{s_m} dt \, \Big[ \Big\| (\cN_+ + 1)^{3/2} e^{i t \phi (h_m)} \prod_{j=m+1}^k  e^{is_j \phi (h_j)} \xi_2 \Big\|  \\ &\hspace{6cm} + \Big\|  (\cN_+ + 1)^{3/2} e^{-i (s_m-t) \left[ \phi (h_m) + D_m \right]} \prod_{j=1}^{m-1}  e^{-is_j \left[ \phi (h_j) + D_j \right]} \xi_2 \Big\| \Big] \end{split} \]
Noticing that $\| h_j \| \leq C \| q_0 O_j \ph_0 \|$ for all $j =1, \dots , k$, and that 
\[ \begin{split} 
\Big\|  (\cN_+ + 1)^{3/2} e^{-i (s_m-t) \left[ \phi (h_m) + D_m \right]} &\prod_{j=1}^{m-1}  e^{-is_j \left[ \phi (h_j) + D_j \right]} \xi_2 \Big\| \\ &= \Big\|  (\cN_+ + 1)^{3/2} T_\eta^* e^{-i (s_m-t) \phi (q_0 O_m \ph_0)} \prod_{j=1}^{m-1}  e^{-is_j \phi (q_0 O_j \ph_0)} T_\eta S T_\tau \Omega \Big\| \end{split} \]
and using (\ref{eq:TNT}), (\ref{eq:SNS}), (\ref{eq:TtNTt}) and Lemma \ref{lm:weyl}, we obtain 
\[  \begin{split} &\left| \langle  \xi_1, e^{is_1 \phi (q_0 O_1 \ph_0)} \dots e^{is_k \phi (q_0 O_k \ph_0)} \xi_1 \rangle - \langle \xi_2, e^{is_1 \phi (h_1)} \dots e^{is_k \phi (h_k)} \xi_2 \rangle \right| 
\\ & \hspace{6cm} \leq \frac{C}{N} \sum_{m=1}^k  \| q_0 O_m \ph_0 \| |s_m| \left( 1 + \sum_{j=1}^k s_j^2 \| q_0 O_j \ph_0 \|^2 \right)^{3/2} \end{split} \]
concluding the proof of (\ref{eq:step2}). 

\medskip

{\it Step 3. Action of $S$.} Let $\xi_3 = T_\tau \Omega$, so that $\xi_2= S \xi_3$. Then
\begin{equation}\label{eq:step3}
\begin{split}  &\left| \langle \xi_2, e^{is_1 \phi (h_1)} \dots e^{i s_k \phi (h_k)} \xi_2 \rangle - \langle \xi_3, e^{is_1 \phi (h_1)} \dots e^{i s_k \phi (h_k)} \xi_3 \rangle \right| \\ &\hspace{6cm} \leq \frac{C}{\sqrt{N}} \sum_{m=1}^k  \| q_0 O_m \ph_0 \| |s_m| \left(1 + \sum_{j=1}^k |s_j|^2 \| q_0 O_j \ph_0 \|^2 \right) \end{split} \end{equation}

\medskip

To show (\ref{eq:step3}), we expand the action of $S = e^{-A}$, with $A$ defined as in (\ref{eq:Adef}). We obtain  
\[ S^* \phi (h_j) S = \phi (h_j) + \int_0^1  e^{\kappa A} [A, \phi (h_j)] e^{-\kappa A} d\kappa =: \phi (h_j) + E_j \]
We have 
\[ \langle \xi_2, e^{is_1 \phi (h_1)} \dots e^{i s_k \phi (h_k)} \xi_2 \rangle  = 
\langle \xi_3, e^{is_1 \left[ \phi (h_1) + E_1 \right]} \dots e^{i s_k \left[ \phi (h_k) + E_k \right] } \xi_3 \rangle\]
and therefore
\[ \begin{split} 
\langle \xi_2, e^{is_1 \phi (h_1)} & \dots e^{i s_k \phi (h_k)} \xi_2 \rangle - \langle \xi_3, e^{is_1 \phi (h_1)} \dots e^{i s_k \phi (h_k)} \xi_3 \rangle \\ = \; &\sum_{m=1}^k \langle \xi_3, \prod_{j=1}^{m-1} e^{is_j \left[ \phi (h_j) + E_j \right]} \left( e^{is_m \left[ \phi (h_m) + E_m \right]} - e^{is_m \phi (h_m)} \right) \prod_{j=m+1}^k e^{is_j \phi (h_j)} \xi_3 \rangle \\
= \; &i \sum_{m=1}^k \int_0^{s_m} dt  \, \langle \xi_3, \prod_{j=1}^{m-1} e^{is_j \left[ \phi (h_j) + E_j \right]} e^{i (s_m -t) \left[ \phi (h_m) + E_m \right]}  E_m e^{it \phi (h_m)} \prod_{j=m+1}^k e^{is_j \phi (h_j)} \xi_3 \rangle \\
= \; &i \sum_{m=1}^k \int_0^{s_m} dt \int_0^1 d\kappa  \, \langle \xi_3, \prod_{j=1}^{m-1} e^{is_j \left[ \phi (h_j) + E_j \right]} e^{i (s_m -t) \left[ \phi (h_m) + E_m \right]}  e^{\kappa A} [A, \phi (h_m)] e^{-\kappa A} \\ &\hspace{9cm} \times e^{it \phi (h_m)} \prod_{j=m+1}^k e^{is_j \phi (h_j)} \xi_3 \rangle
\end{split} \]
With $[A, \phi (h_m)] = [A, b(h_m)] + [A,b^* (h_m)]$ and with Lemma \ref{lemma:commA}, we find 
\[ \begin{split} 
&\left| \langle \xi_2, e^{is_1 \phi (h_1)}  \dots e^{i s_k \phi (h_k)} \xi_2 \rangle - \langle \xi_3, e^{is_1 \phi (h_1)} \dots e^{i s_k \phi (h_k)} \xi_3 \rangle \right| 
\\ &\hspace{1cm} \leq \frac{C}{\sqrt{N}} \sum_{m=1}^k  \| h_m \| \int_0^{s_m} dt  \int_0^1 d\kappa 
\,\Big\| (\cN_+ + 1)^{1/2} e^{-\kappa A} e^{it \phi (h_m)} \prod_{j=m+1}^k e^{is_j \phi (h_j)} \xi_3 \Big\|  \\ 
&\hspace{4cm} \times  \Big\| (\cN_+ + 1)^{1/2} e^{-\kappa A} e^{-i (s_m-t) \left[ \phi (h_m) + E_m \right]} \prod_{j=1}^{m-1} e^{-is_j \left[ \phi (h_j) + E_j \right]} \xi_3 \Big\| \end{split} \] 
Noticing that 
\[ \begin{split} 
\Big\| (\cN_+ + 1)^{1/2} e^{-\kappa A} &e^{-i (s_m-t)  \left[ \phi (h_m) + E_m \right]} \prod_{j=1}^{m-1} e^{-is_j \left[ \phi (h_j) + E_j \right]} \xi_3 \Big\| \\ &= \Big\| (\cN_+ + 1)^{1/2} e^{-\kappa A} S^* e^{-i (s_m-t) \phi (h_m)} \prod_{j=1}^{m-1} e^{-is_j \phi (h_j)} S T_\tau \Omega \Big\| \end{split}  \]
and using (\ref{eq:SNS}), Lemma \ref{lm:weyl} and $\| h_j \| \leq C \| q_0 O_j \ph_0 \|$ for all $j=1,\dots, k$, we arrive at
\[ \begin{split} 
&\left| \langle \xi_2, e^{is_1 \phi (h_1)}  \dots e^{i s_k \phi (h_k)} \xi_2 \rangle - \langle \xi_3, e^{is_1 \phi (h_1)} \dots e^{i s_k \phi (h_k)} \xi_3 \rangle \right| \\ &\hspace{6cm} \leq \frac{C}{\sqrt{N}} \sum_{m=1}^k  \| q_0 O_m \ph_0 \| |s_m| \left(1 + \sum_{j=1}^k |s_j|^2 \| q_0 O_j \ph_0 \|^2 \right) \end{split} \]

\medskip

{\it Step 4. Action of $T_\tau$.} Recall that $\xi_3 = T_\tau \Omega$. We have
\begin{equation}\label{eq:step4} \begin{split}  &\left| \langle \xi_3 , e^{is_1 \phi (h_1)} \dots e^{is_k \phi (h_k)} \xi_3 \rangle - \langle \Omega, e^{is_1 \phi (\nu_1)} \dots e^{is_k \phi (\nu_k)} \Omega \rangle \right| \\ &\hspace{5cm} \leq  \frac{C}{N} \sum_{m=1}^k |s_m| \| q_0 O_m \ph_0 \| \left( 1 + \sum_{j=1}^k s_j^2 \| q_0 O_j \ph_0\|^2 \right)^{3/2}  \end{split}  \end{equation}
where $\nu_j (p) =  \widehat{(q_0 O_j \ph_0)} (p) \cosh (\mu_p) + 
\widehat{(\overline{q_0 O_j \ph_0})} (p) \sinh (\mu_p)$ for $j=1, \dots , k$, and \begin{equation}\label{eq:mup} \mu_p = \frac{1}{4} \log \left( \frac{p^2}{p^2 + 16 \pi \frak{a}_0} \right).\end{equation}
 
 \medskip
 
 From (\ref{eq:act-T}), with $\eta$ replaced by $\tau$, we find, similarly to (\ref{eq:def-h}), 
 \[ T_\tau^* \phi (h_j) T_\tau = \phi (\wt{\nu}_j) + D_j \]
 where
\[ D_j = \sum_{p \in \Lambda^*_+} h_j (p) d_p^* + \overline{h_j (p)} d_p \]
and 
 \[ \begin{split} \wt{\nu}_j (p) = \; &h_j (p) \cosh (\tau_p) + \overline{h_j (-p)} \sinh (\tau_p) 
 \\ = \; &\left[ (\widehat{q_0 O_j \ph_0})(p) \cosh (\eta_p) + (\widehat{\overline{q_0 O_j \ph_0}}) (p) \sinh (\eta_p) \right] \cosh (\tau_p)  \\ &+ \left[ (\widehat{\overline{q_0 O_j \ph_0}})(p) \cosh (\eta_p) + (\widehat{q_0 O_j \ph_0}) (p) \sinh (\eta_p) \right] \sinh (\tau_p) \\
 = \; &(\widehat{q_0 O_j \ph_0})(p) \cosh (\eta_p + \tau_p) + (\widehat{\overline{q_0 O_j \ph_0}})(p) \sinh (\eta_p + \tau_p) \end{split} \]
With (\ref{eq:taup-def}), an explicit computation shows that 
\[ \eta_p + \tau_p = \frac{1}{4} \log \left( \frac{p^2}{p^2 + 2 \, \widehat{Vf_\ell} (p/N)} \right) \]
In particular, this implies that there exists a constant $C > 0$ with $|\eta_p + \tau_p| \leq C$ for all $p \in \Lambda^*_+$. Thus $\| \wt{\nu}_j \| \leq C \| q_0 O_j \ph_0 \|$ for all $j =1, \dots, k$. Proceeding as in Step 2, we therefore obtain 
\begin{equation}\label{eq:step4-1}  \begin{split} & \left| \langle \xi_3 , e^{is_1 \phi (h_1)} \dots e^{is_k \phi (h_k)} \xi_3 \rangle - \langle \Omega, e^{is_1 \phi (\wt{\nu}_1)} \dots e^{is_k \phi (\wt{\nu}_k)} \Omega \rangle \right| \\ &\hspace{6cm} \leq
\frac{C}{N} \sum_{m=1}^k |s_m| \| q_0 O_m \ph_0 \| \left( 1 + \sum_{j=1}^k s_j^2 \| q_0 O_j \ph_0\|^2 \right)^{3/2} \end{split} \end{equation}
Next, we observe that 
\[ |\widehat{V f_\ell} (p/N) - 8 \pi \frak{a}_0 | \leq \left| \int V(x) f_\ell (x) (e^{ip \cdot x /N} - 1) dx  \right| + \left| \int V(x) f_\ell (x) dx - 8\pi \frak{a}_0 \right| \leq \frac{C (|p|+1)}{N} \]
Thus, with $\mu_p$ as defined in (\ref{eq:mup}), we conclude that 
\[ \begin{split} |\eta_p + \tau_p - \mu_p| &= \frac{1}{4} \left| \log \left( \frac{
p^2 + 2 \, \widehat{V f_\ell} (p/N)}{p^2 + 16 \pi \frak{a}_0} \right) \right| \\ &= 
 \frac{1}{4} \left| \log \left( 1 + \frac{2 \widehat{V f_\ell} (p/N) - 16 \pi \frak{a}_0}{p^2 + 16 \pi \frak{a}_0} \right) \right| \\ &\leq \frac{C}{p^2 + 16 \pi \frak{a}_0}  | \widehat{V f_\ell} (p/N) - 8\pi \frak{a}_0| \leq C N^{-1}  \end{split} \]
uniformly in $p$. Therefore, we obtain that 
\[ \begin{split} \| \wt{\nu}_j &- \nu_j \|^2 \\ &\leq 2 \sum_{p \in \Lambda^*_+}  \left[ | \cosh (\eta_p + \tau_p) - \cosh \mu_p |^2 + |\sinh (\eta_p + \tau_p) - \sinh (\mu_p)|^2 \right]  |(\widehat{q_0 O_j \ph_0}) (p)|^2 \\ &= 8 \sum_{p \in \Lambda^*_+} \left[ \cosh^2 \left( \frac{\eta_p + \tau_p + \mu_p}{2} \right) + \sinh^2  \left( \frac{\eta_p + \tau_p + \mu_p}{2} \right) \right] \sinh^2 \left( \frac{\eta_p + \tau_p - \mu_p}{2} \right)  |(\widehat{q_0 O_j \ph_0}) (p)|^2 \\ &\leq \frac{C}{N^2} \| q_0 O_j \ph_0 \|^2 \end{split} \]
Hence, we can compare 
\[ \begin{split} 
\langle \Omega, e^{is_1 \phi (\wt{\nu}_1)} &\dots e^{is_k \phi (\wt{\nu}_k)} \Omega \rangle - \langle \Omega, e^{is_1 \phi (\nu_1)} \dots e^{is_k \phi (\nu_k)} \Omega \rangle \\ &= 
\sum_{m=1}^k \langle \Omega, \prod_{j=1}^{m-1} e^{is_j \phi (\wt{\nu}_j)} \left( e^{is_m \phi (\wt{\nu}_m)} - e^{is_m \phi (\nu_m)} \right) \prod_{j=m+1}^k e^{is_j \phi (\nu_j)} \Omega \rangle \\ &= i\sum_{m=1}^k \int_0^{s_m} dt  \, \langle \Omega, \prod_{j=1}^{m-1} e^{is_j \phi (\wt{\nu}_j)}  e^{i (s_m -t) \phi (\wt{\nu}_m)} \phi (\wt{\nu}_m - \nu_m)  e^{it \phi (\nu_m)} \prod_{j=m+1}^k e^{is_j \phi (\nu_j)} \Omega \rangle
\end{split} \]
and we can estimate, with Lemma \ref{lm:weyl} and since $\| \nu_j \| \leq C \| q_0 O_j \ph_0 \|$ for all $j=1,\dots, k$, 
\[ \begin{split} 
&\left| \langle \Omega, e^{is_1 \phi (\wt{\nu}_1)} \dots e^{is_k \phi (\wt{\nu}_k)} \Omega \rangle - \langle \Omega, e^{is_1 \phi (\nu_1)} \dots e^{is_k \phi (\nu_k)} \Omega \rangle \right| \\ &\hspace{4cm} \leq \sum_{m=1}^k \| \wt{\nu}_m - \nu_m \| \int_0^{s_m} dt  \Big\| (\cN_+ + 1)^{1/2} e^{it \phi (\nu_m)} \prod_{j=m+1}^k e^{is_j \phi (\nu_j)} \Omega \Big\| \\
&\hspace{4cm} \leq \frac{C}{N} \sum_{m=1}^k |s_m| \| q_0 O_m \ph_0 \|  \left(1 + \sum_{j=m}^k s_j^2 \| q_0 O_j \ph_0 \|^2 \right)^{1/2}  \end{split} \]
Combining the last equation with (\ref{eq:step4-1}), we conclude that 
\[ \begin{split} 
& \left| \langle \xi_3 , e^{is_1 \phi (h_1)} \dots e^{is_k \phi (h_k)} \xi_3 \rangle - \langle \Omega, e^{is_1 \phi (\nu_1)} \dots e^{is_k \phi (\nu_k)} \Omega \rangle \right| \\ &\hspace{6cm} \leq \frac{C}{N} \sum_{m=1}^k |s_m| \| q_0 O_m \ph_0 \| \left( 1 + \sum_{j=1}^k s_j^2 \| q_0 O_j \ph_0\|^2 \right)^{3/2} \end{split} \]
as claimed. 

\medskip

{\it Step 5. Replacing modified with standard creation and annihilation operators.} We embed the truncated Fock space $\cF_+^{\leq N}$ into the full Fock space $\cF_+ = \bigoplus_{n \geq 0} L^2_{\perp \ph_0} (\Lambda)^{\otimes_s n}$ constructed on the orthogonal complement of $\ph_0$. On $\cF_+$, we consider the field operators $\phi_a (h) = a^* (h) + a (h)= \sum_{p \in \Lambda^*_+} (h (p) a^*_p + \overline{h(p)} a_p)$. Then, we have 
\begin{equation}\label{eq:step5}\begin{split}   &\left| \langle \Omega, e^{is_1 \phi (\nu_1)} \dots e^{is_k \phi (\nu_k)} \Omega \rangle- 
\langle \Omega, e^{is_1 \phi_a (\nu_1)} \dots e^{is_k \phi_a (\nu_k)} \Omega \rangle \right| \\ &\hspace{5cm} \leq \frac{C}{N} \sum_{m=1}^k |s_m| \| q_0 O_m \ph_0 \|   \, \left( 1 + \sum_{j=1}^m s_j^2 \| q_0 O_j \ph_0 \|^2 \right)^{3/2} \end{split} 
\end{equation}

\medskip

Recalling that $\phi (h) = b(h) + b^* (h) = \sqrt{1-\cN_+/N} \, a (h) + a^* (h) \sqrt{1-\cN_+/N}$, we compute 
\[ \begin{split} 
\langle \Omega, &e^{is_1 \phi (\nu_1)} \dots e^{is_k \phi (\nu_k)} \Omega \rangle- 
\langle \Omega, e^{is_1 \phi_a (\nu_1)} \dots e^{is_k \phi_a (\nu_k)} \Omega \rangle \\ = \; &\sum_{m=1}^k \langle \Omega, \prod_{j=1}^{m-1} e^{is_j \phi (\nu_j)} \left( e^{is_m \phi (\nu_m)} - e^{is_m \phi_a (\nu_m)} \right) \prod_{j=m+1}^k e^{is_j \phi_a (\nu_j)} \Omega 
\rangle \\ = \; &i \sum_{m=1}^k \int_0^{s_m} dt  \, \Big\langle \Omega, \prod_{j=1}^{m-1} e^{is_j \phi (\nu_j)} e^{i (s_m -t) \phi (\nu_m)} \\ &\hspace{1cm} \times \left[ \left( \sqrt{1 - \cN_+/N} - 1\right) a(\nu_m) + a^* (\nu_m) \left( \sqrt{1-\cN_+/N} - 1 \right) \right] e^{i t \phi_a (\nu_m)}  \prod_{j=m+1}^k e^{is_j \phi_a (\nu_j)} \Omega \Big\rangle \end{split} \]
With Lemma \ref{lm:weyl} (and using that $\| \nu_j \| \leq C \| q_0 O_j \ph_0\|$), 
we conclude that 
\[ \begin{split} 
&\left| \langle \Omega, e^{is_1 \phi (\nu_1)} \dots e^{is_k \phi (\nu_k)} \Omega \rangle- 
\langle \Omega, e^{is_1 \phi_a (\nu_1)} \dots e^{is_k \phi_a (\nu_k)} \Omega \rangle \right|  \\ &\hspace{2cm} \leq \; \frac{C}{N} \sum_{m=1}^k \| \nu_m \|  \int_0^{s_m} dt  \, \Big\| (\cN_+ +1)^{3/2}  e^{-i (s_m-t) \phi (\nu_m)} \prod_{j=1}^{m-1} e^{-i s_j \phi (\nu_j)} \Omega \Big\| 
\\ &\hspace{2cm} \leq \; \frac{C}{N} \sum_{m=1}^k |s_m| \| q_0 O_m \ph_0 \|   \, \left( 1 + \sum_{j=1}^m s_j^2 \| q_0 O_j \ph_0 \|^2 \right)^{3/2}  \end{split} \]

\medskip

{\it Step 6.} We have
\begin{equation}\label{eq:step6} \langle \Omega, e^{is_1 \phi_a (\nu_1)} \dots e^{is_k \phi_a (\nu_k)} \Omega \rangle = e^{-\frac{1}{2} \sum_{\ell,j=1}^k s_\ell s_j \Sigma_{\ell, j}} \end{equation}
with the $k \times k$ matrix 
\[ \Sigma_{\ell,j} = \left\{ \begin{array}{ll} \langle \nu_\ell, \nu_j \rangle, \quad &\text{if $\ell \leq j$} 
\\ \langle \nu_j, \nu_\ell \rangle, \quad &\text{if $j < \ell$} \end{array} \right. \]
 
\medskip

From the Weyl relation
\begin{align*}
e^{i \phi_a (f)} e^{i \phi_a (g)}  = e^{i \phi_a (f+g)} e^{- i \mathrm{Im} \langle f,g \rangle}
\end{align*}
for all $f,g \in L^2( \Lambda )$, we obtain, setting $\nu = \sum_{j=1}^k s_j \nu_j$, 
\begin{align*}
\prod_{j =1}^k e^{i s_j \phi_a (\nu_j)} = e^{i \phi_a (\nu)} \prod_{\ell<j}^k e^{- i s_\ell s_j \mathrm{Im} \langle \nu_\ell, \nu_j \rangle} = e^{-\| \nu \|^2/2} \prod_{\ell<j}^k e^{- i s_\ell s_j \mathrm{Im} \langle \nu_\ell, \nu_j \rangle} e^{i a^* (\nu)} e^{i a (\nu)}.
\end{align*}
Thus
\[ \langle  \Omega, e^{is_1 \phi_a (\nu_1)} \dots e^{is_k \phi_a (\nu_k)} \Omega \rangle = 
e^{-\| \nu \|^2/2} \prod_{\ell<j}^k e^{- i s_\ell s_j \mathrm{Im} \langle \nu_\ell, \nu_j \rangle} = e^{-\frac{1}{2} \sum_{\ell,j=1}^k s_\ell s_j \Sigma_{\ell, j}} \]
as claimed. 

\medskip

Combining now the results of (\ref{eq:step1}), (\ref{eq:step2}), (\ref{eq:step3}), (\ref{eq:step4}), (\ref{eq:step5}) and (\ref{eq:step6}), and inserting in (\ref{eq:pro1}), we conclude that 
\[ \begin{split} &\left| \mathbb{E}_{\psi_N}\left[ g_1(\mathcal{O}_{1}) \dots g_k ( \mathcal{O}_{k}) \right] - \int ds_1 \dots ds_k \; \widehat{g}_1(s_1) \dots \widehat{g}_k( s_k) \;  e^{-\frac{1}{2} \sum_{\ell,j=1}^k s_\ell s_j \Sigma_{\ell, j}} \right| \\ &\hspace{8cm} \leq \frac{C}{N^{1/4}} \prod_{j=1}^k \int ds \; | \widehat{g}_j (s)| (1 + |s|^4 \| O_j \|^4) \end{split} \]
for a constant $C > 0$ depending on $k$. Taking the Fourier transform, we conclude (under the assumption that the matrix $\Sigma$ is invertible) that 
\[  \begin{split} &\left| \mathbb{E}_{\psi_N}\left[ g_1(\mathcal{O}_{1}) \dots g_k ( \mathcal{O}_{k}) \right] - \int d\lambda_1 \dots d\lambda_k \, g_1(\lambda_1) \dots g_k (\lambda_k) \;  \frac{e^{-\frac{1}{2} \sum_{\ell,j=1}^k \Sigma^{-1}_{\ell, j} \lambda_\ell \lambda_j}}{\sqrt{(2\pi)^k \det \Sigma}}  \right| \\ &\hspace{8cm} \leq \frac{C}{N^{1/4}} \prod_{j=1}^k \int ds \; | \widehat{g}_j (s)| (1 + |s|^4 \| O_j \|^4) \end{split} \]
which concludes the proof of Theorem \ref{thm:main}
\end{proof}

\textit{Acknowledgements.} B.\,S.\ acknowledges support from the Swiss National Science Foundation (grant 200020\_172623) and from the NCCR SwissMAP.

\appendix

\section{Excited States}
\label{app}

Theorem \ref{thm:main} and Corollary \ref{cor:BE} describe the probability distribution associated with the ground state wave function $\psi_N$. From the analysis of \cite{BBCS}, we also obtain norm-approximations for eigenvectors associated to low-energy excited states of the Hamilton operator (\ref{eq:HN}). It follows from \cite[Section 6]{BBCS} that the eigenvector $\psi_N^{(p)}$ describing a single excitation of the ground state, with momentum $p \in \Lambda^*_+$ can be approximated (assuming non-degeneracy) by 
\begin{equation}\label{eq:exc-p} \| \psi_N^{(p)} - e^{i\omega} U_N^* T_\eta S T_\tau a_p^* \Omega \| \leq C N^{-1/4} \end{equation}
for an appropriate phase $\omega \in \bR$. 
Using (\ref{eq:exc-p}), we can study the probability distribution associated with $\psi_N^{(p)}$ as well. For $k \in \bN$, bounded one-particle observables $O_1, \dots , O_k$, functions $g_1, \dots , g_k \in L^1 (\bR)$ with Fourier transform $\widehat{g}_1, \dots , \widehat{g}_k \in L^1 (\bR, (1+|s|^4) ds)$, we obtain, repeating the analysis of Section \ref{sec:proof}, 
\[ \begin{split} &\left| \mathbb{E}_{\psi_N^{(p)}}  \left[ g_1 (\cO_1) \dots g_k (\cO_k) \right]  - \int ds_1 \dots ds_k \; \widehat{g} (s_1) \dots \widehat{g}_k (s_k) \, \langle a_p^* \Omega , \prod_{j=1}^k e^{is_j \phi_a (\nu_j)} a_p^* \Omega \rangle \right|  \\ &\hspace{8cm} \leq \frac{C}{N^{1/4}} \prod_{j=1}^k \int ds \, |\widehat{g} (s)| (1 + |s|^4 \| O_j \|^4) \end{split} \]
Setting $\nu (p) = \sum_{j=1}^k s_j \nu_j$, we obtain that 
\[  \begin{split}  \langle a_p^* \Omega , \prod_{j=1}^k e^{is_j \phi_a (\nu_j)} a_p^* \Omega \rangle &= \prod_{\ell < j} e^{- i s_\ell s_j \text{Im } \langle \nu_\ell , \nu_j \rangle} \langle a_p^* \Omega, e^{i \phi_a (\nu)} a_p^* \Omega \rangle \\ &= e^{-\| \nu \|^2/2} \prod_{\ell < j} e^{- i s_\ell s_j \text{Im} \langle \nu_\ell , \nu_j \rangle} \langle a_p^* \Omega, e^{i a^* (\nu)} e^{i a(\nu)} a_p^* \Omega \rangle \\ &= e^{-\frac{1}{2} \sum_{\ell , j = 1}^k \Sigma_{\ell,j} s_\ell s_j}  \langle (a_p^* + i \nu (p)) \Omega, (a_p^* - i \nu (p)) \Omega \rangle  \\ &= e^{-\frac{1}{2} \sum_{\ell , j = 1}^k \Sigma_{\ell,j} s_\ell s_j} (1+ |\nu (p)|^2) =  e^{-\frac{1}{2} \sum_{\ell , j = 1}^k \Sigma_{\ell,j} s_\ell s_j} \Big( 1 + \big| \sum_{j=1}^k s_j \nu_j (p) \big|^2 \Big)
\end{split} \]
Thus, we find  
\[ \begin{split} &\left| \mathbb{E}_{\psi_N^{(p)}}  \left[ g_1 (\cO_1) \dots g_k (\cO_k) \right]  - \int ds_1 \dots ds_k \; \widehat{g} (s_1) \dots \widehat{g}_k (s_k) \, e^{-\frac{1}{2} \sum_{\ell , j = 1}^k \Sigma_{\ell,j} s_\ell s_j} \Big(1 +  \big| \sum_{j=1}^k s_j \nu_j (p) \big|^2 \Big) \right| \\
&\hspace{10cm} \leq  \frac{C}{N^{1/4}} \prod_{j=1}^k \int ds \, |\widehat{g} (s)| (1 + |s|^4 \| O_j \|^4) \end{split} \]
Analogously, through more complicated combinatorics, we could also describe the probability distribution associated with states having multiple excitations.

\end{document}